\documentclass{amsart}
\usepackage[left=1in,right=1in,top=1.2in,bottom=1in]{geometry}
\usepackage[utf8]{inputenc}
\usepackage{setspace}
\usepackage{amsmath,amssymb}
\usepackage{comment}
\usepackage{xcolor}
\usepackage{url}
\usepackage[colorlinks=true,allcolors=blue,backref=page]{hyperref}
\usepackage[noabbrev,capitalize,nameinlink]{cleveref}
\usepackage{bbm}

\newtheorem{thm}{Theorem}
\newtheorem{lem}[thm]{Lemma}
\newtheorem{cor}[thm]{Corollary}

\newtheorem{fact}[thm]{Fact}

\theoremstyle{remark}

\newtheorem{example}{Example}

\theoremstyle{definition}

\newtheorem{algorithm}{Algorithm}

\newcommand{\R}{\mathbb{R}}

\newcommand{\eps}{\varepsilon}
\renewcommand{\epsilon}{\eps}
\renewcommand{\P}{\mathbb{P}}

\newcommand{\E}{\mathbb{E}}
\newcommand{\cE}{\mathcal{E}}

\newcommand{\Z}{\mathbb{Z}}
\newcommand{\mast}{{m_\ast}}

\newcommand{\vs}{v_{\mathrm{s}}}
\newcommand{\dist}{\mathrm{dist}}

\begin{document}
\title{Age of gossip from connective properties via first passage percolation}

\author{Thomas Jacob~Maranzatto}

\address{University of Maryland. Department
of Electrical and Computer Engineering}
\email{tmaran@umd.edu}

\author{Marcus Michelen}
\address{University of Illinois, Chicago. Department of Mathematics, Statistics and Computer science.}
\email{michelen@uic.edu}

\begin{abstract}
    In gossip networks, a source node forwards time-stamped updates to a network of observers according to a Poisson process.  The observers then update each other on this information according to Poisson processes as well. 
    The Age of Information (AoI) of a given node is the difference between the current time and the most recent time-stamp of source information that the node has received.

    We provide a method for evaluating the AoI of a node in terms of first passage percolation.  We then use this distributional identity to prove matching upper and lower bounds on the AoI in terms of connectivity properties of the underlying network.  In particular, if one sets $X_v$ to be the AoI of node $v$ on a finite graph $G$ with $n$ nodes, then we define $m_\ast = \min\{m : m \cdot |B_m(v)| \geq n\}$ where $B_m(v)$ is the ball of radius $m$ in $G$.  In the  case when the maximum degree of $G$ is bounded by $\Delta$ we prove $\E X_v = \Theta_\Delta(m_\ast)$.  
    As corollaries, we solve multiple open problems in the literature such as showing the age of information on a subset of $\Z^d$ is $\Theta(n^{1/(d+1)})$.  
    
    We also demonstrate examples of graphs with AoI scaling like $n^{\alpha}$ for each $\alpha \in (0,1/2)$.  These graphs are not vertex-transitive and in fact we show that if one considers the AoI on a graph coming from a vertex-transitive infinite graph then either $\E X_v = \Theta(n^{1/k})$ for some integer $k \geq 2$ or $\E X_v = n^{o(1)}$.
\end{abstract}

\maketitle

\section{Introduction}
Given a social network of gossiping people, how out of date is a given person's information?  While there are many such possible models for this problem, the \emph{Age of Information} (AoI) metric has garnered a great deal of interest in the information theory community since its introduction by Kaul, Yates, and Gruteser~\cite{aoiseminal}. 

The setting for defining the AoI metric in a gossip model is a weighted finite directed graph $G = (V,E,W)$.  Introduce an additional node $\vs$ referred to as the \emph{source} and add a directed edge from $\vs$ to each $v \in V$ and additional weights to these arcs.  Throughout all times $t \geq 0$ for each $v \in V$ we will maintain a non-negative real number $N_v(t)$ denoting the timestamp of the most recent packet $v$ has seen; we initialize\footnote{As we will primarily be interested in the $t \to \infty$ limit, the choice of initialization is not too important.} with $N_v(0) = 0$ for all $v \in V$.  For all $t$ we have $N_{\vs}(t) = t$.   On each directed edge $e$, we have mutually independent Poisson clocks with rate $w_e$ where $w_e$ is the weight of $E$ in $G$.  When the Poisson process on an edge of the form $(u,v)$  rings at time $t$, then one updates $N_v(t) = \max\{N_v(t),N_u(t)\}$.  That is, the \emph{fresh} packet is kept at vertex $v$.  The \emph{Age of Information} is then defined by $X_v(t) := t - N_v(t)$.

This process may be interpreted as the source node broadcasting the state of the world at time $t$ and all other nodes updating each other dynamically, with the quantity $X_v(t)$ representing how out of date the information is that $v$ has at time $t$.  A special case of weight choices is in the case of Yates' model \cite{yates_gossip_isit} where one assumes that every node has the same fixed communication capacity $\lambda$ and so for each edge $e = (u,v)$ we set the weight $w_{(u,v)} := \lambda / \mathrm{outdeg}(u)$.  In particular, for each non-source $v \in V$ we have $w_{(\vs,v)} = \lambda/n$.  In the case when the underlying graph begins as an undirected graph, we replace each edge with two directed edges, one in each direction.  We refer to this choice of weights as \emph{Yates weights}.  We note that in general $w_{uv} \not= w_{vu}$.

Much work~\cite{Buyukates_2022, gossip_dynamic,   kaswan2023age, Kaswan_2023, maranzatto2024age, gossip_grid,  yates_gossip_isit, yates_gossip_spawc}   has focused on analyzing the limiting average AoI defined by $\overline{X}_v := \lim_{t \to \infty} \E X_v(t)$ in various specific cases of the underlying graph $G$.  A main tool in the area is the Stochastic Hybrid System (SHS) approach.  This was originally used by Yates \cite[Theorem 1]{yates_gossip_isit}, who also proves the existence of the limit defining $\overline{X}_v$. The recursive SHS formula \cite[Theorem 1]{yates_gossip_isit} sets $X_S(t) := \min_{v \in S} X_v(t)$ and $\overline{X}_S := \lim_{t\to \infty} \E X_S(t)$; it writes $\overline{X}_S$ in terms of the quantities $\overline{X}_{S \cup u}$ for all $u$ that neighbor the set $S$.  As such, obtaining bounds on $\overline{X}_v$ through this approach requires unraveling a recursion that controls $\overline{X}_S$ for all connected subsets $S$ containing $v$.

Our contribution is twofold.  First, we obtain a distributional identify for the random variables $X_v(t)$ in the case of arbitrary weights in terms of \emph{first passage percolation} on the same graph but with all edges reversed.  As an immediate corollary, we obtain not only the convergence of the limit of expectations $\lim_{t \to \infty}\E X_v(t)$ but also convergence in distribution of the random variable $X_v(t)$ as $t \to \infty$.  Further, this exact formula provides a description for each finite time $t$ in addition to the $t \to \infty$ case.  This is stated precisely in \cref{thm:age_distribution}.  In the language of Markov processes, we show that first passage percolation is the \emph{dual process} of the Markov process defining the AoI variables $X_v(t)$.  Roughly, this means that if one reverses time, then the process remains Markovian and is exactly a version of first passage percolation.

Second, using \cref{thm:age_distribution} we obtain upper and lower bounds on $\E X_v(t)$ in terms of geometric qualities of the underlying graph $G$ when we use Yates weights.  To set up this theorem, in an undirected graph $G$ let $\dist_G(u,v)$ denote the graph distance, i.e.\ the least number of edges in a path from $u$ to $v$.  Define $B_m(v) := \{u : \dist_G(u,v) \leq m\}$ to be the ball of radius $m$ centered at $v$.  We now define the quantities $$m_\ast(v) := \min\{m : |B_m(v)| \cdot m \geq |V(G)|\}\qquad \text{ and } \qquad \Phi_m(v) := \frac{|B_m(v)|}{|B_{m-1}(v)|}\,.$$
Finally, we write $\delta$ to be the minimum degree of $G$ and $\Delta$ to be the maximum degree of $G$.

\begin{thm} \label{thm:expected_age}
    For every $v \in V(G)$ with Yates weights we have $$\frac{1}{20}\min\left\{\frac{\delta}{\Delta}, \frac{1}{\Phi_{\mast(v)}(v)}, \frac{t}{\mast(v)}\right\} \leq \frac{\E {X}_v(t)}{\lambda \mast(v)} \leq 3 \Delta\,. $$
\end{thm}

We note that if $t$ is sufficiently large then this shows that $\E X_v(t) = \Theta_\Delta(\lambda \mast(v))$.
\cref{thm:expected_age} immediately implies many bounds on the limiting AoI in many instances.  In the case of a cycle on $n$ nodes, an upper bound of $O(\sqrt{n})$ was proven by Buyukates, Bastopcu, and Ulukus~\cite{Buyukates_2022}.  Srivastava and Ulukus upper bounded the AoI on $\Z_m^2$ with $n:=m^2$ nodes by $O(n^{1/3})$~\cite{gossip_grid}, and conjectured that on $\Z_m^d$ one has the bound of $O(n^{1/(d+1)})$.  Results for the parameterized class of `generalized cycles' were given by the same authors~\cite{gossip_rings}.  The first author analyzed the AoI on random regular graphs \cite{maranzatto2024age}.
\cref{thm:expected_age} reproves these results, proves the conjecture for the $d$-dimensional lattice and proves corresponding lower bounds that are sharp up to constants.  

The first open question in the survey \cite{kaswan2023age} states that previous results ``[hint] at a synergy between
connectivity and freshness: the more connected a network is,
the fresher its nodes will be through gossiping. Describing this synergy in its entirety by finding age scalings of networks with medium-level connectivities is an important open problem.''  From this perspective, \cref{thm:expected_age} provides a direct relationship between connectivity and AoI.  In \cref{sec:examples} we outline many examples of applications of \cref{thm:expected_age}.  

Seeking ``medium-level'' connectivities as discussed in \cite{kaswan2023age}, we show that for each $\alpha \in (0,1/2)$ one can construct a connected graph $G$ of bounded degree so that $\overline{X}_v = \Theta(n^\alpha)$ (this is \cref{ex:intermediate}).  The graph chosen depends on the parameter $n$, but we in fact show that this is in some sense a requirement.  If one takes $G$ by taking portions of an infinite vertex-transitive graph then we show that one has either $\overline{X}_v = \Theta(n^{1/k})$ for an integer $k \geq 1$ or $\overline{X}_v = n^{o(1)}$ (handled in \cref{ex:poly-growth} and \cref{ex:big-growth} respectively).  In particular, this shows that for $\alpha$ satisfying $1/\alpha \notin \Z$ it is impossible to obtain a graph on which one has $\overline{X}_v = \Theta(n^\alpha)$ by taking a ball from a vertex-transitive graph.  We prove this by combining \cref{thm:expected_age} with statements from geometric group theory by Gromov and Trofimov in \cref{ex:poly-growth}.

In \cref{subsection:lattice}, we analyze the Hamming cube to show that the lower bound in \cref{thm:expected_age} is sharp up to an absolute constant (i.e.\ independent of $d$).  Similarly, we analyze the regular tree in \cref{subsection:tree} to show that the upper bound in \cref{thm:expected_age} is sharp up to an absolute constant.

Finally, we handle the case of random geometric graphs in \cref{ex:RGG}, random regular graphs in \cref{ex:RRG} and a cycle plus a random matching in \cref{ex:CRM}.

\subsection{Related work on gossip}

Gossiping protocols have existed in network and database applications since at least the 1980's, and have proven successful in many different problems~\cite{gossip_membership, gossip_promise, database_gossip, gossip_dynamic}.  In the basic model, multiple servers in a distributed network aim to keep updated information by passing messages, and importantly these messages are sent independently by the servers with no coordination.  The principle behind this protocol is that when fresh data becomes available to a server by some external process, the dissemination of this data will not take `too long' to reach the other servers when compared to a tailor-made protocol.  Furthermore the gossiping algorithm is simple to implement, so real-world applications might benefit from the ease of deployment over small gains in data freshness.

Recently there has been considerable theoretical interest in quantifying the freshness of information in gossip networks through the lens of Stochastic Hybrid Systems (SHS).  This direction was initiated by Yates~\cite{yates_gossip_isit} to study the first-order behavior of the age of information in fully-connected gossip networks.  The key result from this work is using the SHS framework to translate the stochastic quantity of interest into a recurrence which can be solved or bounded with algebraic manipulations.  Our work provides an alternative methodology to study gossip models, and contrasts by providing a complete characterization of the age distribution.

We note that some related work has focused on the \textit{version Age of Information} (vAoI) where packets are marked by positive integer version numbers, not real-valued timestamps.  The source generates versions by a separate Poisson process with rate $\lambda_e$, and the vAoI is the difference between the source and node version.  Yates~\cite{yates_gossip_isit, yates_gossip_spawc} showed the limiting average vAoI and AoI are the same up to a scalar multiple.  We note that \cref{thm:age_distribution} also applies for vAoI by simply generating a Poisson random variable of mean $\lambda_e X_v(t)$ to obtain the distribution of the vAoI at finite time $t$.  As such, we focus on AoI since the case of vAoI is analogous.

In~\cite{yates_gossip_isit}, it was shown that the vAoI of the complete graph $K_n$ scales as $\Theta(\log n)$, and on the disconnected graph $\overline{K_n}$ the vAoI scales as $\Theta(n)$.  A conjecture was also given that the cycle $C_n$ should have version age $O(\sqrt{n})$ which was subsequently proven by Buyukates, Bastopcu, and Ulukus~\cite{Buyukates_2022}.  Srivastava and Ulukus show the vAoI of the 2D lattice grows as $O(n^{1/3})$~\cite{gossip_grid}, and scaling results for the parameterized class of `generalized cycles' were given by the same authors~\cite{gossip_rings}.  Random graphs and Bipartite graphs have been considered by the first author~\cite{maranzatto2024age}.

We refer the reader to a general survey on AoI~\cite{Yates_survey} and a survey specialized to gossiping networks~\cite{kaswan2023age} for more background.

\subsection{First passage percolation}
We now define first passage percolation and formally state our connection between AoI and first passage percolation (\cref{thm:age_distribution}).  Let $G = (V, E, W)$ be a finite weighted directed graph.  Typically, the weights in $W = \{\tau_e\}_{e \in E}$ are chosen according to some probability distribution, often independently.  A \textit{path} $\gamma$ is a sequence of edges $e_1,e_2,...,e_m$ in $E$ such that for every $i \in [m-1]$, the starting node of $e_i$ is the final node of $e_{i+1}$. The \textit{passage time} of $\gamma$ is
\[T(\gamma) := \sum_{e \in \gamma} \tau_e. \]
Then for any two nodes $u,v \in V$, the \textit{passage time} from $u$ to $v$ is
\[T(u,v) := \min_\gamma  T(\gamma),\]
where the minimum is over paths with $u$ as the starting node and $v$ as the final node.  The passage time between $u$ and $v$ can be interpreted as a distortion or perturbation of the usual graph distance; in particular, if one sets $\tau_e = 1$ for all $e \in E$ then one recovers the usual graph distance.  
Most classic examples in the first passage percolation literature consider the case when $\{\tau_e\}_{e \in E}$ are independent and identically distributed.  

We now set up our connection between the AoI $X_v(t)$ and first passage percolation.  Let $G = (V, E, W)$ be a gossip network with source $\vs$.  We define the \textbf{auxiliary graph} $G' = (V, E', W')$ on the same set of nodes. \begin{itemize}
    \item The edges are $E' = \{(u,v) : (v,u) \in E\}$, i.e.\ all edges are reversed in $G'$.
    \item The weights $W' = (\tau_e)_{e \in E'}$ are independent exponential random variables where the rate of $\tau_{(j,i)}$ is $w_{(i,j)}$.
\end{itemize} 

The graph $G'$ is simply $G$ with all edges inverted, and the Poisson processes on edges are replaced with exponential random variables with an appropriate parameter.  We show the age of a node is essentially the first passage time from $v$ to $\vs$ in $G'$.

\begin{thm}\label{thm:age_distribution}
    For each time $t$ and node $v$, the age process of $v$ at time $t$ has distribution given by \[X_v(t) = \min\{ 
 T(v,\vs),t\}\,.\]
\end{thm}

The main idea for the proof of \cref{thm:age_distribution} is to work backwards from time $t$.  For instance, if one  asks for how much time has passed since a given neighbor of $v$ has updated to $v$, then this is an exponential variable.  We then may ask which neighbor of $v$ most recently sent an update to $v$, and continue onward backwards in time.  Due to the fact that a time-reversed Poisson process is a Poisson process and the memorylessness properties of the exponential random variables, we may reveal the delays along edges at any given point in time and still obtain exponential random variables of the appropriate parameter.  Exploring the shortest sequence of delays backwards from $v$ to $\vs$ in this iterative fashion is exactly the same as running Dijkstra's algorithm with edge weights given by independent exponential random variables, i.e.\ the same as computing $T(v,\vs)$ in first passage percolation.  We prove this connection in \cref{sec:fpp}.  We also note that if one considers multiple source nodes $\{\vs^{(j)}\}_{j}$ then our proof similarly proves $X_v(t) = \min\{ \min_j T_{G'}(v,\vs^{(j)}), t\}$, i.e.\ one needs to look at the shortest passage time to \emph{any} source rather than just the single source.  For simplicity we restrict to a single source throughout.

First passage percolation was first introduced in 1965 by Hammersley and Welsh \cite{hammersley1965first} as a model for how a liquid may travel through a disorded porous medium.  It is also a basic model of the spread of disease and is occasionally referred to as the SI model, since each node is either ``Susceptible'' or ``Infected.''  One can interpret first passage percolation as understanding the spread of a single packet of information through a network.  There has been much work on first passage percolation in probability theory, with the most classic settings being the case of the undirected lattice $\Z^d$ with $d \geq 2$ and the weights $(\tau_e)_{e \in E}$ being independent and identically distributed.  In the special case of i.i.d.\ exponential random variables and $d = 2$, the model is equivalent to the Eden model \cite{eden1958probabilistic,eden1961two}, an even older growth model.   Growth models of this form are also sometimes called a Richardson model, named as such due to Richardson's proof of a shape theorem for the Eden model \cite{richardson1973random}.  For more context on first passage percolation we refer the reader to the survey \cite{auffinger2017fifty}. 

We note that the present work is far from the first to note a connection between gossip models and percolation.  For instance, works of Aldous \cite{aldous2013knowing} and Selen, Nazarathy, Andrew, and Vu \cite{selen2013age} explicitly make this connection in the motivations of their work.  A key novelty in our case is that this connection can be made exact in the model at hand, even though more than one packet of information is sent by the source.  The analysis of this model by the reversal of time appears to be a new ingredient for this model, hence the novelty of \cref{thm:age_distribution} and \cref{thm:expected_age} within the literature on the AoI since its introduction in \cite{aoiseminal}.

\section{Proof of \cref{thm:age_distribution}}\label{sec:fpp}

In order to prove \cref{thm:age_distribution} we will explore the most recent update path to $v$ from $\vs$ node by node and see that this is precisely Dijsktra's algorithm for evaluating the distance $T(v,\vs)$ in first passage percolation.  As such, we review Dijkstra's algorithm now, which for a given node $v_0$ will output the distances $T(v_0,u)$ for all $u \in V$.

\begin{algorithm}[Dijkstra's algorithm] \label{algorithm}  Given a directed weighted graph $G$ and node $v_0 \in V(G)$ perform the following.
\begin{enumerate}
    \item Initially, all vertices are marked as \emph{unvisited} and we let $U$ be the set of unvisited nodes.
    \item For each $u \neq v_0$ we initialize the distance to $v_0$ to be $+\infty$ and set $T(v_0,v_0) = 0$.
    \item \label{step:unvisited} From the set of unvisited nodes, set $u$ to be the node with least distance to $v_0$ and set it to be the \emph{current node}.  If no unvisited nodes remain, terminate the algorithm and output $\{T(v_0,v)\}_{v \in V(G)}$.
    \item For the current node $u$, enumerate its neighbors $\{x_1,\ldots,x_\ell\}$ and update $$T(v_0,x_j) = \min\{T(v_0,x_j), T(v_0,u) + \tau_{u,x_j}\}$$ for each $j \in [\ell]$. \label{step:weight}
    
    \item Mark the current node $u$ as visited and remove it from the unvisited set.  Return to step \eqref{step:unvisited}.
\end{enumerate}
\end{algorithm}

Crucially, we note that in first passage percolation with exponential edge weights, we may choose to only reveal the random variable $\tau_e$ when it is explored in step \eqref{step:weight}.  In particular, first passage percolation with exponential edge weights may equivalently be interpreted as the following model instead.   Given a weighted directed graph $G = (V,E,W)$ on each edge $e$ place a Poisson process of intensity $w_e$.  For two vertices $v_0,u \in V$ define the random variable $\widetilde{T}(v_0,u)$ via 
\begin{equation}\label{eq:wT-def}
    \widetilde{T}(v_0,u) := \min\left\{ \sum_{j \geq 0} s_j : \begin{array}{c}\exists~v_0,v_1,\ldots,v_k \text{ with } v_k = u \text{ and } s_j \geq 0  \text{ so that } \\
   \sum_{i \geq j} s_i \text{ is an update time of } (v_j,v_{j+1}) \text{ for all }j \in \{0,\ldots,k-1\}\end{array}\right\}\,.
\end{equation}

We now show that this quantity is exactly distributed as the passage time $T(v_0,u)$ with exponential variables placed on each edge.

\begin{lem}\label{lem:first-passage-dynamic}
    For a directed weighted graph $G(V,E,W)$ the distribution of the quantity $\widetilde{T}(v_0,u)$ in \eqref{eq:wT-def} is equal to the distribution of $T(v_0,u)$ in first passage percolation where one places independent exponential random variables of intensity $w_e$ on each edge $\tau_e$. 
\end{lem}
\begin{proof}
    First note that we may in fact write $$\widetilde{T}(v_0,u) = \min\left\{ \sum_{j \geq 0} s_j : \begin{array}{c}\exists~v_0,v_1,\ldots,v_k \text{ with } v_k = u\, \\
   s_j = \min\{ s : s \geq 0, \sum_{i \geq j+1} s_i - s \text{ is an update time of }(v_j,v_{j+1})\}\end{array}\right\}$$
   as taking $s_j$ to be the minimum such delay cannot increase the sum $\sum_j s_j$.  In order to identify $\widetilde{T}(v_0,u)$ and $T(v_0,u)$, we will couple the exponential random variables $\tau_{e}$ in $T(v_0,u)$ with the Poisson clocks used for $\widetilde{T}$ to see that the two are equal.  In particular, we will find $T$ by running Dijkstra's algorithm (\cref{algorithm}) and revealing the randomness in $\tau_e$ only in step \eqref{step:weight}.  

   Suppose we have completed step \eqref{step:unvisited} of \cref{algorithm}.  When we move to step \eqref{step:weight}, let $u$ be the current node and $x_1,\ldots,x_\ell$ be its neighbors.  If we define $$\tau_{u,x_j} = \min\{ s : s \geq 0, T(v_0,u) + s \text{ is an update time of }(u,x_j)\}  $$
   then we see that the variables $\tau_{u,x_j}$ are mutually independent of each other and all previously revealed times $(\tau_e)$ by the strong Markov property of the Poisson process.  Further, each $\tau_{u,x_j}$ is an exponential random variable of rate $w_{u,x_j}$.  We may then show by induction on the number of iterations taken in Dijkstra's algorithm that this maintains $\widetilde{T}(v_0,u) = T(v_0,u)$ for this coupling, completing the proof. 
\end{proof}

With this reinterpretation in mind, \cref{thm:age_distribution} follows quickly.

\begin{proof}[Proof of \cref{thm:age_distribution}]
    Note that we may write \begin{equation*}
    X_v(t) = t - \max\left\{ t_0 : \begin{array}{c}\exists~v_0,v_1,\ldots,v_k \text{ with }v_0 = \vs, v_k = v \text{ and } \\
    0 \leq t_0 \leq t_1 \leq \ldots \leq t_{k-1} \leq t \text{ so that } \\
    t_j \text{ is an update time of }(v_j,v_{j+1}) \text{ for all }j \in \{0,\ldots,k-1\}\end{array}\right\}
    \end{equation*}
    where we set the maximum to be $0$ if no such $t_0$ exists. 
    For an arbitrary vertex $u \in V$, we may define  \begin{equation}\label{eq:AoI-optimization}
    X_{u \to v}(t) = t - \max\left\{ t_0 : \begin{array}{c}\exists~v_0,v_1,\ldots,v_k \text{ with }v_0 = u, v_k = v \text{ and } \\
    0 \leq t_0 \leq t_1 \leq \ldots \leq t_{k-1} \leq t \text{ so that } \\
    t_j \text{ is an update time of }(v_j,v_{j+1}) \text{ for all }j \in \{0,\ldots,k-1\}\end{array}\right\}
    \end{equation}
    and so we have $X_v(t) = X_{\vs \to v}(t)$.  We will switch to considering the delays between the $t_j$'s rather than the times themselves.  With this in mind, if we change variables and set $s_{k - 1} = t - t_{k-1}$ and $s_j = t_{j+1} - t_j$ for $j \in \{0,1,\ldots,k-2\}$ then we can rewrite \eqref{eq:AoI-optimization} as 
    \begin{align}
    X_{u \to v}(t) &= \min\left\{ \sum_{j \geq 0} s_j : \begin{array}{c}\exists~v_0,v_1,\ldots,v_k \text{ with }v_0 = u, v_k = v \text{ and } \\
    s_j \geq 0 \text{ for } j \in \{0,1,\ldots,k-1\}, \sum_{j \geq 0} s_j \leq t \text{ so that } \\
    t - \sum_{i \geq j} s_i \text{ is an update time of } (v_i,v_{i+1}) \text{ for all }i \in \{0,\ldots,k-1\}\end{array}\right\} \nonumber 
    \end{align}
    Since the time reversal of a homogeneous Poisson process is again a Poisson process of the same intensity, \cref{lem:first-passage-dynamic} completes the proof. 
\end{proof}

\section{A general bound for expected age in gossip networks}\label{sec:expectation}
We specialize our study to the model of gossip networks introduced by Yates~\cite{yates_gossip_isit, yates_gossip_spawc}.  Throughout we fix a communication graph $G = (V, E)$ with source $\vs$. In this model, there is an arc between the source and every node.  For all $v \in V$ the weight of an edge $(v,u)$ is $\frac{\lambda}{\mathrm{outdeg}(v)}$. We will assume $\lambda = 1$ without loss of generality, as we can otherwise rescale $X_v$ by $\lambda$.

Let $B_m := \{u : \operatorname{dist}_G(u,v) \leq m\}$ be the ball of radius $m$ about $v$, where $\operatorname{dist}_G(\cdot, \cdot)$ is the unweighted graph distance.  Define $\mast := \min \{ m : |B_m(v)| m \geq n\}$ and $\Phi_m := \frac{|B_m(v)|}{|B_{m-1}(v)|}$. Recall $\delta$ is the minimum degree of $G$ and $\Delta$ is the maximum degree of $G$. 

Before we prove \cref{thm:expected_age} we require two bounds on sums of i.i.d.\ exponential random variables.  We use the following ``small ball probability bound'' which appears in \cite[Exercise 2.2.10(b)]{vhdp}:

\begin{lem}\label{lem:reverse_tail}
        For $Y_1,\ldots,Y_N$ i.i.d. exponential random variables with parameter $\mu$,
    \[
        \P\left( \sum_{j = 1}^N Y_j \leq t\right) \leq \left(\frac{et\mu}{N}\right)^{N}\,.
    \]
\end{lem}

We also need the following upper tail bound which appears in \cite[Theorem 5.1(i)]{janson2018tail}:
\begin{lem}\label{lem:upper_tail}
        For $Y_1,\ldots,Y_N$ i.i.d. exponential random variables with parameter $1$ we have \begin{equation*}
            \P\left(\sum_{j = 1}^N Y_j \geq t N \right) \leq \frac{1}{t}\exp\left(- n(t - 1 - \log t)\right)
        \end{equation*}
        for all $t \geq 1$.
\end{lem}

\begin{proof}[Proof of \cref{thm:expected_age}] Throughout the proof write $\mast = \mast(v), \Phi_\mast = \Phi_{\mast(v)}(v)$ and $B_r = B_r(v)$ for all $r \geq 0$.
We begin with the lower bound.  By \cref{thm:age_distribution}, $\P(X_v(t) \geq a) = \P(T(v,\vs) \geq a)$ for $a \leq t$ and $0$ otherwise, so we focus our attention on the graph used for first passage percolation $G'$. Let $K := \frac{1}{10}\min\{\frac{\delta}{\Delta},\frac{1}{\Phi_\mast},\frac{t}{\mast} \}.$    Define the event $\cE_1 := \{\forall u \in B_\mast, \tau_{(u, \vs)} \geq K\mast\}$.  Then,
\begin{align*} 
\P(\cE_1) &= \exp \left( \frac{-K \mast |B_\mast|}{n} \right) = \exp \left( \frac{-K \mast |B_{\mast-1}| \Phi_{\mast}}{n} \right) \geq \exp \left( -\frac{2K(\mast - 1)|B_{\mast - 1}| \Phi_\mast}{n} \right)\\
&\geq \exp \left(-2K \Phi_\mast \right) \geq \exp(-1/5)\,.
\end{align*}
Now consider the event $$\cE_2 := \{\forall~u \in B_\mast \setminus B_{\mast - 1}, \forall~\gamma \text{ from }v \text{ to }u, T(\gamma) \geq K\mast \}\,.$$  

In words, $\cE_2$ is the event no short path exists from $v$ to the boundary of $B_{\mast}$.  For any node $x \in B_{\mast} \setminus B_{\mast - 1}$, the number of paths of length $L$ from $x$ to $v$ is at most $\Delta^L$.  Let $\{Y_i\}_1^L$ be i.i.d.\ exponential random variables with rate $1/\delta$. By \cref{lem:reverse_tail} we may bound 
\begin{align*}
\P(T(\gamma) \leq K\mast) &\leq \P\left(\sum_{i=1}^L Y_i \leq K\mast\right) \leq \left(\frac{ e K\mast}{L\delta} \right)^L 
\end{align*}
Then summing over all paths,

\begin{align*}
    \P(\cE_2^c) &\leq \sum_{L \geq \mast} \left(\Delta \frac{e K\mast}{\delta L} \right)^L \leq \sum_{k \geq 1} \left( \frac{e }{10} \right)^k = \frac{e/10}{1-e/10}\,.
\end{align*}
Further, $\cE_1$ and $\cE_2$ are independent and on the event $\cE_1 \cap \cE_2$ we have $ X_v(t) \geq K\mast$.  We then obtain a lower bound $$\E X_v(t) \geq \P(\cE_1)\P(\cE_2) K \mast \geq \frac{\mast}{20} \cdot \min\left\{\frac{\delta}{\Delta}, \frac{1}{\Phi_\mast}, \frac{t}{\mast} \right\}\,. $$

To prove the upper bound, analogous to the lower bound, we focus on $G'$.  First for a parameter $a >0$ define $\cE_1 := \left\{\exists~u \in B_{\mast} : \tau_{u,\vs} \leq a \right\}$\,.  We may then bound
\begin{align}\label{eq:cE1c}
    \P(\cE_1^c) = \P\left( \forall~u \in B_{\mast} : \tau_{u,\vs} \geq a   \right) = \P\left( \mathrm{Exp}(1/n) \geq a\right)^{|B_\mast|} = \exp\left(- a |B_{\mast}| / n\right) \leq \exp\left(-a/\mast\right)
\end{align}
where the first equality is because we are using Yates weights, and the last inequality is since $\mast |B_\mast| \geq n$.

We now note that for any $u \in B_{\mast}$ there is some path $\gamma$ from $u$ to $v$ using at most $m_\ast$ edges.  If we let $Y_j$ be i.i.d.\ exponential variables of mean $1$ and write $a = \Delta \mast \theta$ for $\theta \geq 1$ then we see by \cref{lem:upper_tail},
\begin{equation}\label{eq:upper-tail-path}
    \P(T(\gamma) \geq a) \leq \P\left(\sum_{j = 1}^\mast Y_j \geq \theta \mast \right) \leq \frac{1}{\theta}\exp\left(-\mast (\theta - 1 - \log \theta) \right)\,.
\end{equation}

We then see that under the event $\cE_1$ for $a \geq \Delta m_\ast$ we have that there is some vertex $u \in B_\mast$ with $\tau_{u,\vs} \leq a$; letting $\gamma$ denote any path from $u$ to $v$, we note that $T(\gamma)$ is independent of $\tau_{u,\vs}$, and so the upper bound in \eqref{eq:upper-tail-path} holds.  Combining \eqref{eq:upper-tail-path} with \eqref{eq:cE1c} provides \begin{align*}
    \E X_v(t) &\leq \Delta \mast + \int_{\Delta \mast}^\infty \P(X_v(t) \geq a) \,da \\
    &\leq \Delta \mast + \int_{\Delta \mast}^\infty e^{- a / \mast } \,da + \int_{\Delta \mast }^\infty \P(T(\gamma) \geq a) \,da \\
    &\leq \Delta \mast + \mast + \Delta \mast \int_1^\infty \frac{1}{\theta} e^{-(\theta - 1 - \log \theta)} d\theta \\&\leq 3 \Delta \mast  
\end{align*}
 thus completing the proof.
\end{proof}

\section{Applications}
\label{sec:examples}

\begin{example}[Transitive graphs of polynomial growth] \label{ex:poly-growth}
    Let $\mathcal{G}$ be a vertex-transitive infinite graph and let $v_0 \in \mathcal{G}$.  Assume that $\mathcal{G}$ has polynomial growth, i.e.\ there are constants $C,\kappa > 0$ so that $|B_r(v_0)| \leq C r^\kappa$ for all $r > 0$.   Work of Gromov and Trofimov imply that in fact there is a constant $C > 0$ and \emph{integer} $d$ so that 
    \begin{equation}\label{eq:poly-growth}
        \frac{1}{C} r^d \leq |B_r(v_0)| \leq C r^d 
    \end{equation}
    i.e.\ balls must grow like an integer power (see the discussion in \cite{imrich1989survey} in the paragraph below Theorem 2.2). 
    
    Let $G = B_r(v_0)$, i.e.\ the ball of radius $r$ in the graph $\mathcal{G}$ centered at $v_0$.  Let $n = |G|$ and note that \eqref{eq:poly-growth} shows that $m_\ast(v) = \Theta(n^{1/(d+1)})$ thus showing $\E X_{v_0}(t) = \Theta(n^{1/(d+1)})$ by \cref{thm:expected_age}.  In particular, for well-behaved connected subgraphs of vertex-transitive graphs of polynomial growth, the AoI is asymptotically $\Theta(n^{1/k})$ where $k\geq 2$ is some integer.
\end{example}

\begin{example}[Transitive graphs of superpolynomial growth]\label{ex:big-growth}
    In the setting of \cref{ex:poly-growth}, if one assumes that $\mathcal{G}$ has superpolynomial growth---i.e.\ that for each $A$ one has that $|B_r(v_0)| \geq r^A$ for all sufficiently large $r$---then one obtains that $\E X_{v_0}(t) \leq n^{o(1)}$.  We note if $\mathcal{G}$ has exponential growth\footnote{It is possible for $\mathcal{G}$ to have intermediate growth, i.e.\ superpolynomial growth but subexponential growth, e.g.\ in the case of a Cayley graph of the Grigorchuk group.} then $\E X_{v_0}(t) = \Theta( \log n)$. 
\end{example}

\begin{example}[Graphs with arbitrary polynomial power AoI]  \label{ex:intermediate} Fix $\alpha \in (0,1/2)$.   
Consider the random $3$-regular graph $H$ on $m$ nodes and replace each edge with a path of length $\lfloor m^\gamma \rfloor$ to obtain a graph $G$.  The graph $G$ has $n := (3/2 + o(1))m^{1 + \gamma}$ nodes, maximum degree $3$ and minimum degree $2$.  Since $H$ has diameter $O(\log m)$ with high probability \cite{bollobas_diameter_1982}, the graph $G$ has diameter $O(m^\gamma \log m)$.  Further, for each natural number $k$, each ball of radius $k \lfloor m^\gamma\rfloor$ in $G$ has at most $3^k m^\gamma$ nodes.  This implies that for each $v$ we have that $m_\ast(v) = \Theta( n^\gamma \log n)$ since we are approximately solving the equation $$(k m^\gamma) 3^k m^\gamma = n =(1 + o(1))m^{1 +\gamma}$$
and obtain a solution of $k = \Theta(\log m)$.  We then have that the AoI satisfies $$\E X_t = {\Theta}_\gamma(m^{\gamma} \log m)\,.$$
Take $\gamma \in (0,1)$ to solve the equation \begin{equation*}
    \frac{\log \log m}{\log m} + \gamma = (1 + \gamma)\alpha
\end{equation*}
which is always possible for $m$ large enough as a function of $\alpha \in (0,1/2)$. We then obtain that 
$$\E X_t = {\Theta}_\alpha(m^{\gamma} \log m)= {\Theta}_\alpha(n^\alpha)\,. $$  We note that in light of \cref{ex:poly-growth}, all such examples must not come from transitive graphs.
\end{example}

\begin{example}[The integer lattice] \label{item:lattice} If $\Z_m^d$ is the $d$-dimensional toroidal lattice with $n = m^d$, then $\mast = \Theta_d( n^{1/(d+1)} )$, so  $\E X_t = \Theta_d(n^\frac{1}{d+1})$.  Note also that by \cref{ex:poly-growth} if one takes $G$ to be the ball of radius $r$ in $\Z^d$ then the same result holds as well. 
\end{example}

\begin{example}[Random geometric graphs] \label{ex:RGG}
Let $G = G(\gamma, n;d)$ be the $d$-dimensional random geometric graph, where nodes are vectors $x \in \R^d, \|x\|_2 < 1$ distributed uniformly at random in the unit ball and $(u,v) \in E$ if and only if $\|u -v\|_2 < \gamma$.  If we parameterize $\gamma = \alpha\left(\frac{\log n}{V_d n}\right)^{1/d}$ for $\alpha > 1$ and $V_d$ denoting the volume of the unit ball, then $G$ is connected\footnote{
To get a feeling for connectivity occurring above this quantity, note that that for $\gamma = \alpha (\frac{\log n}{V_d n})^{1/d}$, the average degree is asymptotically $\alpha^d \log n$.  If one parameterizes the Erd\H{o}s-R\'enyi random graph in terms of average degree, $\log n$ is also the threshold for connectivity.} with high probability \cite{godehardt1996connectivity,penrose2003random}.  We also have that for each $r > 0$ if one considers a random point in $v \in G$ then one has $$\E B_r(v) \leq n V_d (r\gamma)^d = (\alpha \log n r)^d\,.$$

Further, a shape theorem for the distances in continuum percolation \cite{yao2011large} implies that with high probability we also have $$B_r(v) = \widetilde{\Omega}_{\alpha,d}(r^d)\,.$$

In particular, this implies that $\mast = \widetilde{\Theta}(r^{1/(d+1)})$ with high probability.  Letting $\Delta$ denote the maximum degree, a union bound along with a Chernoff bound\footnote{Here we use the bound $\P(\mathrm{Bin}(n,p) \geq (1 + \lambda) np ) \leq \exp(-c \lambda np)$ for some universal $c > 0$.} shows \begin{align*}
	\P(\Delta \geq C \alpha^d \log^2 n) \leq n \P\left(\mathrm{Bin}(n, \alpha^{d}\frac{\log n}{n}) \geq C \alpha^d \log^2 n  \right) \leq n \cdot \exp\left(- \frac{c \cdot C n \alpha^d \log^2 n}{\alpha^d n \log n } \right) \leq \frac{1}{n}
\end{align*}
for a sufficiently large universal constant $C > 0$.  Thus by \cref{thm:expected_age} we have $\E X_v(t) = \widetilde{\Theta}_{\alpha,d}(n^{1/(d+1)})$, which matches the behavior of the lattice in \cref{item:lattice}. This  demonstrates that random geometric graphs in a metric space behave in a similar fashion to lattices.  In particular, in this regime their AoI is governed by global geometric properties of the space rather than local properties.
\end{example}

\begin{example}[Random regular graphs] \label{ex:RRG}
For constant $d$, let $G$ be the random $d$-regular graph.  Bollob\'as and Vega~\cite{bollobas_diameter_1982} showed the diameter of $G$ is $O(\log(n \log n))$, so $X_t = \Theta(\log n)$. This is another proof of the same result by the first author~\cite{maranzatto2024age}.
\end{example} 

\begin{example}[Cycle plus random matching] \label{ex:CRM}
If $G = C + \mathcal{M}$, is the union of a cycle and a random matching, then a classic result by Bollob\'as and Chung~\cite{cycle_matching} gives the diameter of $G$ is with high probability $O(\log n)$, therefore $X_t = \Theta(\log n)$.  This shows that adding $O(n)$ edges can reduce the version age of a graph from $O(\sqrt{n})$ to asymptotically optimal, and one can interpret the random matching as highways along which information can travel between notes that are far in the underlying graph $C$.
\end{example}

\section{Sharpness of Theorem \ref{thm:expected_age}: regular trees and the Hamming cube}\label{sharpness}

In a setting where the maximum degree of the graph is bounded, the upper and lower bounds in \cref{thm:expected_age} are equal up to a constant depending on this maximum degree.  Here we provide more detailed analysis to show that both the upper and lower bounds are in fact sharp up to universal constants for some networks.  For the lower bound we will look at the Hamming cube and for the upper bound we will look at regular trees.

\subsection{Sharpness of the lower bound: the Hamming cube} \label{subsection:lattice}

Let $Q_d := \{0,1\}^d$ under nearest neighbor distances.  We will first upper bound $\E X(t)$ and then analyze $\Phi_{\mast}$ to see that the lower bound in \cref{thm:expected_age} meets this bound.

\begin{thm}\label{thm:hypercube}
There is a universal constant $C > 0$ so that in the Hamming cube $Q_d$ we have $$\E X(t) \leq C d\,.$$
\end{thm}

To begin with, we will deduce from previous work on percolation an upper bound on the expected first passage time (see, e.g., \cite{fill-pemantle}).

\begin{fact}\label{fact:percolation-hyper}
    Orient $Q_d$ so that each edge points towards the node of larger hamming weight.  For each $\eps > 0$ there is a constant $c(\eps) > 0$ so that the following holds.  If we keep each edge independently with probability $(1 + \eps)e/d$ and delete it with the remaining probability, then the probability that there is an oriented path from $(0,0,\ldots,0)$ to $(1,1,\ldots,1)$ is bounded below by $c$.  
\end{fact}

From here, an upper bound on oriented first-passage percolation with follows quickly.

\begin{cor}\label{cor:hypercube}
    In oriented first-passage percolation on $Q_d$ with i.i.d.\ standard exponential edge weights, there is a universal constant so that $$\max_{x,y} \E T(x,y) \leq C$$
    provided there is a directed path from $x$ to $y$.
\end{cor}
\begin{proof}
    Without loss of generality, suppose $x = (0,0,\ldots,0)$ and $y = (1,1,\ldots,1)$ are opposite ends of the hypercube.  Write $\P_p$ to denote the probability measure given by performing $p$ percolation on the oriented hypercube, and let $x \to y$ be the event that $x$ is connected to $y$.  We first claim that for any $p \in (0,1)$ and integer $k$ we have \begin{equation}\label{eq:indep-dominate}
    \P_{1-(1-p)^k}( x \to y) \geq 1 - (1 - \P_{p}(x \to y))^k\,.
    \end{equation}
    We may see this by noting that if we sample $p$-percolation $k$ times independently, then the probability a given edge is open in at least one instance is exactly $1 - (1 - p)^k$, and so the union of these open sets is exactly $1 - (1 - p)^k$ percolation; if $x \to y$ in any of the independent $p$ percolations then $x \to y$ in the union, thus showing \eqref{eq:indep-dominate}.  If we take $p = 3/d$ then we note that for, e.g., $k \leq d/100$ we have $1 - (1 - 3/d)^k \leq \frac{10k}{d}\,.$  Using \cref{fact:percolation-hyper} we deduce the bound $$\P_{\lambda/d}(x \to y) \geq 1 - 2\exp(-c \lambda)$$
    for some constant $c > 0$. Note that for $\lambda \leq d/100$, we have that the probability a standard exponential random variable is less than $2\lambda/d$ is at least $\lambda/d$. 
    If there is a directed path with each exponential on that path being at most $2\lambda/d$, then the passage time is at most $2\lambda$.  We then see $$\P(T(x,y) \geq \lambda) \leq 1 - \P_{\lambda /(2d)}(x \to y) \leq 2 \exp(-c \lambda /2)$$
    thus completing the proof.
\end{proof}

Much more is known about $T(x,y)$ in the setting of \cref{cor:hypercube}.  In particular, if $x$ and $y$ are opposite points of the hypercube, then Fill and Pemantle \cite{fill-pemantle} show that $T(x,y) = 1 + o(1)$ with high probability.  We include the argument in \cref{cor:hypercube}  because we require a bound on $\E T(x,y)$ and not just a bound on $T(x,y)$ with high probability.

\begin{proof}[Proof of \cref{thm:hypercube}]
    Note that all weights $\tau_e$ to non-source nodes are exponential random variables of weight $d$.  The expected time until a packet it sent from the source is $1$, and the expected time for this to reach a given vertex is $O(d)$ by \cref{cor:hypercube}.  Applying \cref{thm:age_distribution} completes the proof.
\end{proof}

We now analyze the quantity in the lower bound of \cref{thm:expected_age} in the context of $Q_d$.   First note that $\delta = \Delta = d$ since $Q_d$ is regular.  We will also see that $\Phi_\mast = 1 + o(1)$, but require an estimate on the size of balls in the hamming cube first.  We use the following standard estimate which appears, for instance in \cite[Theorem 5.11]{spencer-book}:
\begin{fact}\label{fact:hamming-ball}
    Set $r = \frac{d}{2} - \alpha\sqrt{d}$ for $0 < \alpha = o(d^{1/6})$ and $\alpha = \omega(1).$  Then
    $$|B_r| = (1 + o(1))\cdot \frac{2^d }{\alpha \sqrt{8\pi}}\exp(-2 \alpha^2)\,.$$
\end{fact}
Asymptotics for $\Phi_\mast$ follow quickly:
\begin{cor}
    In the Hamming cube if we parameterize $\mast = \frac{d}{2} - \alpha \sqrt{d}$ then $\alpha = (1 + o(1))\sqrt{ \frac{\log d}{2}}\,.$  Further, we have $\Phi_{\mast} = 1 + o(1)\,.$
\end{cor}
\begin{proof}
    The asymptotic identity for $\alpha$ follows from \cref{fact:hamming-ball}.  For $\Phi_\mast$, note that if we write $\mast - 1 = \frac{d}{2} - \alpha'$ then we have $\alpha' = \alpha + \frac{1}{\sqrt{d}}$ and so \cref{fact:hamming-ball} implies \begin{equation*}
    \Phi_\mast = \frac{|B_\mast|}{|B_{\mast - 1}|} = (1 + o(1)) \frac{\alpha'}{\alpha} \exp(-2\alpha^2 + 2( \alpha')^2) = (1 + o(1))\exp\left(O(\sqrt{\log d / d})\right) = 1 + o(1)\,. \qedhere  \end{equation*}
\end{proof}

\subsection{Sharpness of the upper bound: regular trees} \label{subsection:tree}
Consider the $\Delta$-regular tree on $n$ nodes ($n$ chosen appropriately).  Note that for the root node $v_0$ we have $\mast(v_0) = (1 + o(1)) \log_\Delta(n)$, so the following result matches the upper bound in \cref{thm:expected_age}.
\begin{thm}
    There is some constant $n_0 := n_0(\Delta)$ such that for the root $v_0$ of the $\Delta$-regular tree with $n \geq n_0$ nodes, the average expected age of $v_0$ is at least $\frac{\Delta}{7} \log_\Delta n$
\end{thm}
\begin{proof}
    We will use \cref{thm:age_distribution} and show that $$\P\left( X_{v_0}(t) \geq \frac{\Delta \log_\Delta n}{6}\right) \geq 1 - o(1)$$ which will then complete the proof.  Let $G$ denote the graph $G'$ from  \cref{thm:age_distribution}.  Recalling that $\vs$ is the source node, define the random set 
    $$U := \left\{v \in V : \tau_{(v,\vs)} \leq \frac{\Delta \log_\Delta n}{6} \right\}\,.$$
    Define the (rare) event $$\mathcal{B} := \left\{ \min_{u \in U} \mathrm{dist}(u,v_0) \leq \frac{\mast}{2} \right\}\,.$$  A union bound then shows \begin{align*}
        \P(\mathcal{B}) \leq \sum_{v \in B_{\mast / 2}(v_0)} \P(\tau_{(v,\vs)} \leq \mast / 2) \leq |B_{\mast / 2}(v_0)|  \cdot \frac{\Delta \log_\Delta n}{n} \leq \frac{\Delta \log_\Delta n}{\sqrt{n}}(1 + o(1))
    \end{align*}
    which tends to $0$ as $n \to \infty$.  For a given node $v \in V$ set $d = \mathrm{dist}(v,v_0)$.  If $d \geq \mast / 2$ then compute \begin{align*}
        \P\left(T_G(v_0,v) \leq \frac{\Delta \log_\Delta n}{6} \right) \leq \left(\frac{e \log_\Delta n}{6 d} \right)^d \leq \left(\frac{e}{3}\right)^{\log_\Delta n / 2}\,.
    \end{align*} 
    We may then bound 
    \begin{align*}
       \P\left(X_{v_0}(t) \leq \frac{\Delta \log_\Delta n}{6} \right) &\leq \P(\mathcal{B}) + \P\left(\exists~u \in U, T_G(v_0,u) \leq \frac{\Delta \log_\Delta n}{6} \cap  \mathcal{B}^c\right) \\
        &\leq o(1) + \sum_{ \substack{u \in G \\  \mathrm{dist}(u,v_0) \geq \mast / 2}} \P(u \in U)\P\left(T_G(v_0,v) \leq \frac{\Delta \log_\Delta n}{6} \right) \\
        &\leq o(1) n \cdot \frac{\Delta \log_\Delta n}{n} \cdot \left(\frac{e}{3} \right)^{\log_\Delta n / 2} \\
        &= o(1)
    \end{align*}
    thus completing the proof.    
\end{proof}

\section*{Acknowledgments}
MM is supported in part by NSF CAREER grant DMS-2336788 as well as grants DMS-2137623 and DMS-2246624.
JM is supported in part by NSF Grant ECCS-2217023.

\bibliographystyle{plain}
\bibliography{bibliography}

\end{document}